\documentclass[journal,onecolumn]{IEEEtran}

\usepackage{amsmath,amsthm,mathtools}

\usepackage{graphicx}
\usepackage{color}
\usepackage{epsfig, amsmath, amsthm, amssymb,latexsym,amsmath,amsbsy, rotating}
\usepackage{amsfonts}

\newtheorem{thm}{Theorem}
\newtheorem{lemma}[thm]{Lemma}
\newtheorem{cor}{Corollary}

\newtheorem{remark}{Remark}

\theoremstyle{definition}

\newcommand\YM{\ensuremath{Y(M)}}
\newcommand\ZM{\ensuremath{Z(M)}}

\newcommand\ZA{\ensuremath{Z({\cal A})}}
\newcommand\ZAc{\ensuremath{Z({\cal A}^c)}}
\newcommand\YA{\ensuremath{Y({\cal A})}}
\newcommand\YAc{\ensuremath{Y({\cal A}^c)}}

\newcommand\sA{\ensuremath{\sigma^2_{\cal A}}}
\newcommand\sAc{\ensuremath{\sigma^2_{{\cal A}^c}}}

\DeclarePairedDelimiter\parentheses{\lparen}{\rparen}
\newcommand{\entp}[1]{\operatorname{N} \parentheses*{#1}}

\DeclareFontFamily{U}{futm}{}
\DeclareFontShape{U}{futm}{m}{n}{
  <-> s * [.92] fourier-bb
  }{}
\DeclareSymbolFont{Ufutm}{U}{futm}{m}{n}
\DeclareSymbolFontAlphabet{\mathbb}{Ufutm}

\begin{document}

\title{Remote Source Coding under Gaussian Noise :\\Dueling Roles of Power and Entropy Power}

\author{Krishnan~Eswaran~and~Michael~Gastpar
\thanks{The work in this manuscript was supported in part by the U.S. National Science Foundation under award CCF-0347298 (CAREER), CNS-0326503, by a fellowship from the Vodafone U.S. Foundation (to K.E.), and by EPFL. The work in this manuscript was partially presented at the {\it 2005 IEEE International Symposium on Information Theory,} Adelaide, Australia, the {\it 39th Conference on Information Sciences and Systems,} The Johns Hopkins University, MD, USA,  the {\it 2006 IEEE Information Theory Workshop}, Punta del Este, Uruguay, and the {\it 2018 Workshop on Information Theory and Applications (ITA)}, San Diego, CA, USA.}
\thanks{K. Eswaran and M. Gastpar were with the Department of Electrical Engineering and Computer Sciences at the University of California, Berkeley. K. Eswaran is now with Google, Inc, and M. Gastpar is with the School of Computer and Communication Sciences, {\'E}cole Polytechnique F{\'e}d{\'e}rale de Lausanne (EPFL), Lausanne, Switzerland (email: michael.gastpar@epfl.ch).}
}


\maketitle

\begin{abstract}
The distributed remote source coding (so-called CEO) problem is studied in the case where the underlying source, not necessarily Gaussian, has finite differential entropy
and the observation noise is Gaussian.
The main result is a new lower bound for the sum-rate-distortion function under arbitrary distortion measures.
When specialized to the case of mean-squared error, it is shown that the bound exactly mirrors a corresponding upper bound,
except that the upper bound has the source power (variance) whereas the lower bound has the source entropy power.
Bounds exhibiting this pleasing duality of power and entropy power have been well known for direct and centralized source coding since Shannon's work.
While the bounds hold generally, their value is most pronounced when interpreted as a function of the number of agents in the CEO problem.
\end{abstract}


\IEEEpeerreviewmaketitle

\section{Introduction}

In the CEO problem, there is an underlying source and $M$ encoders~\cite{GelfandP:79,FlynnG:87,bergerzhangvis-it96}. Each encoder gets a noisy observation of the underlying source. The encoders provide rate-limited descriptions of their noisy observations to a central decoder. The central decoder produces an approximation of the underlying source to the highest possible fidelity.
This work studies the special case where the observation noise is additive Gaussian and independent between different encoders.
When the underlying source is also Gaussian and the fidelity criterion is the mean-squared error, this problem is referred to as the {\it quadratic Gaussian CEO problem} and is well studied in the literature~\cite{viswanathan-it97,oohama-it98,oohama-it05,prabharakaran-isit04,TavildarVW:10}.
In the work presented here, we still consider additive Gaussian observation noises, but we allow the underlying source to be any continuous distribution constrained to having a finite differential entropy. We refer to this as the {\it AWGN CEO problem}.
The contributions of the work are the following:
\begin{itemize}
\item A new general lower bound is presented for the AWGN CEO problem with an arbitrary underlying source, not necessarily Gaussian, and subject to an arbitrary distortion measure. (Theorems~\ref{thm-CEO-aux} and~\ref{thm-AWGN-CEO-generallowerbound}.)
\item When specialized to the case of the mean-squared error distortion measure, the new lower bound is shown to closely match a known upper bound. In fact, both bounds assume the same shape, except that the lower bound has the entropy power whereas the upper bound has the source power (variance). This parallels the well-known Shannon lower bound for the standard rate-distortion function under mean-squared error. (Corollaries~\ref{Cor-CEO-aux-MMSE} and~\ref{thm-AWGN-CEO}.)
\item The strength of the new bounds is that they reflect the correct behavior as a function of the number of agents $M.$ This fact is leveraged and illustrated in two follow-up results. The first characterizes the rate loss in the CEO problem, {\it i.e.,} the rate penalty of distributed versus centralized encoding (Theorem~\ref{lemma-priceofdistributed}). The second pertains to a network joint source-channel coding problem (more specifically, a simple model of a sensor network), given in Theorem~\ref{thm-jscc}.
\end{itemize}
The underpinnings of the new bounds leverage and extend work by Oohama~\cite{oohama-it05}, by Wagner and Anantharam \cite{wagnerthesis-05,WagnerA:08} and by Courtade~\cite{Courtade:it18}.

We also note that there is a wealth of work about further versions of the CEO problem. Strategies are explored in~\cite{ChenB:it08}.
The case of so-called log-loss is addressed in~\cite{CourtadeW:14}.
There is also an interesting connection between the CEO problem and the problem of so-called ``nomadic'' communication and oblivious relaying, where one strategy is for intermediate nodes to compress their received signals~\cite{SanderovichSSK:it08,DBLP:journals/corr/AguerriZCS17}.

\subsection*{Notation}

All logarithms in this paper are natural, and $\log^+x = \max\{ 0, \log x \}.$
Random variables will be denoted by upper case letters $U.$
Random vectors will be denoted by boldface upper case letters $\mathbf{U} = (U_1, U_2, \cdots, U_M).$
For every subset ${\cal A} \subseteq \{1, 2, \cdots, M\},$ we will use $\mathbf{U}_{\cal A}$ to denote the subset of those components of $\mathbf{U}$ whose indices are in ${\cal A}.$
Moreover, ${\cal A}^c$ denotes the complement of the set ${\cal A}$ in $\{1, 2, \cdots, M\}.$
Given a random variable $X$ with density $f_X(x)$,
its variance is denoted by $\sigma_X^2,$
its differential entropy is $h(X) = -\int f_X(x) \log f_X(x) dx,$
and its entropy power is
\begin{align}
\entp{X} & = \frac{e^{2h(X)}}{2 \pi e},
\end{align}
and we recall that for Gaussian random variables $X,$ we have that $\entp{X} = \sigma_X^2.$
Finally, we will use the notation $X \leftrightarrow Y \leftrightarrow Z$ to denote Markov chains, i.e., the statement that $X$ and $Z$ are conditionally independent given $Y.$

\section{CEO Problem Statement}\label{sec-problem-CEO}

\subsection{The CEO Problem}

The CEO problem is a standard problem in multi-terminal information theory.
For completeness, we include a brief formal problem statement here.
An underlying source is modeled as a string $X^n$ of length $n$ of independent and identically distributed (i.i.d.) continuous random variables $\{X\}_{i=1}^n,$
following the terminology in~\cite[p.~243]{CoverThomas06}.
Throughout this study, we assume that the corresponding entropy power $\entp{X}$ is non-zero and finite.
The source $X$ is observed by $M$ encoding terminals through a broadcast channel $f_{Y_1, Y_2, \ldots, Y_M|X}(y_1, y_2, \ldots, y_M|x).$
The observation sequences $Y_m^n,$ $m = 1, 2, \ldots, M,$ are separately encoded with the goal of finding an estimate $\hat{X}^n$ of $X^n$ with distortion $D.$

\begin{figure}
  \begin{center}
  \setlength{\unitlength}{1.5pt}
  \begin{picture}(140,130)(130,0)
\thicklines
    \put (130, 60) {\line (1, 0) {15} }
    \put (133, 63) {\makebox (10,10) {$X^n$}}

    \put (145, 10) {\line (0, 1) {105} }

    \put (145, 115) {\vector (1, 0) {5} }
    \put (150, 115) {\line (1, 0) {20} }
    \put (155, 115) {\circle {10} }
    \put (155, 110) {\line (0, 1) {10} }
    \put (155, 130) {\vector (0, -1) {10} }
    \put (158, 122) {\makebox (10,10) {$Z_1^n$}}

    \put (145, 80) {\vector (1, 0) {5} }
    \put (150, 80) {\line (1, 0) {20} }
    \put (155, 80) {\circle {10} }
    \put (155, 75) {\line (0, 1) {10} }
    \put (155, 95) {\vector (0, -1) {10} }
    \put (158, 87) {\makebox (10,10) {$Z_2^n$}}

    \put (145, 10) {\vector (1, 0) {5} }
    \put (150, 10) {\line (1, 0) {20} }
    \put (155, 10) {\circle {10} }
    \put (155, 5) {\line (0, 1) {10} }
    \put (155, 25) {\vector (0, -1) {10} }
    \put (158, 17) {\makebox (10,10) {$Z_M^n$}}

    \put (170, 115) {\vector (1, 0) {19} }
    \put (175, 118) {\makebox (10,10) {$Y_1^n$}}
    \put (189, 105) {\framebox (25,20) {\sc enc 1} }
    \put (214, 115) {\line (1, 0) {3} }
    \put (219, 115) {\line (1, 0) {3} }
    \put (226, 115) {\line (1, 0) {3} }
    \put (228, 115) {\vector (1, 0) {5} }
    \put (214, 118) {\makebox (19,10) {$nR_1$}}

    \put (170, 80) {\vector (1, 0) {19} }
    \put (175, 83) {\makebox (10,10) {$Y_2^n$}}
    \put (189, 70) {\framebox (25,20) {\sc enc 2} }
    \put (214, 80) {\line (1, 0) {3} }
    \put (219, 80) {\line (1, 0) {3} }
    \put (226, 80) {\line (1, 0) {3} }
    \put (228, 80) {\vector (1, 0) {5} }
    \put (214, 83) {\makebox (19,10) {$nR_2$}}

    \put (200, 50) {\makebox (10,10) {$\vdots$}}

    \put (170, 10) {\vector (1, 0) {19} }
    \put (175, 13) {\makebox (10,10) {$Y_M^n$}}
     \put (189, 0) {\framebox (25,20) {\sc enc $M$} }
    \put (214, 10) {\line (1, 0) {3} }
    \put (219, 10) {\line (1, 0) {3} }
    \put (226, 10) {\line (1, 0) {3} }
    \put (228, 10) {\vector (1, 0) {5} }
    \put (214, 13) {\makebox (19,10) {$nR_M$}}

    \put (233, 0) {\framebox (20,125) {\sc dec} }
    \put (253, 60) {\vector (1, 0) {19} }
    \put (258, 63) {\makebox (10,10) {${\hat X}^n$}}
  \end{picture}
  \end{center}
\caption{The $M$-agent AWGN CEO problem. $X$ is an arbitrary source with variance (power) $\sigma_X^2$ (not necessarily Gaussian) and entropy power $\entp{X}.$ The observation noises $Z_i$ are independent and Gaussian.}
\label{fig-CEO}
\end{figure}
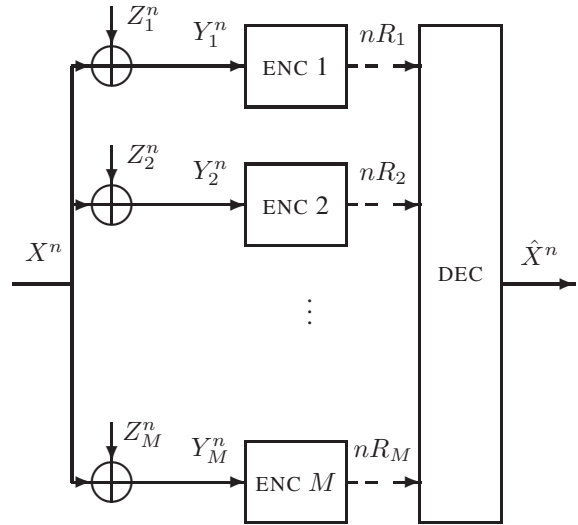

A $(2^{nR_1}, 2^{nR_2}, \cdots, 2^{nR_M}, n)$ code for the CEO problem consists of
\begin{itemize}
\item $M$ encoders, where encoder $m$ assigns an index $j_m(y_m^n) \in \{ 1, 2, \cdots, 2^{nR_m} \}$ to each sequence $y_m^n,$
for $m=1, 2, \cdots, M,$ and
\item a decoder that assigns an estimate $\hat{x}^n$ to each index tuple $(j_1, j_2, \cdots, j_M).$
\end{itemize}
A rate-distortion tuple $(R_1, R_2, \cdots, R_M, D)$ is said to be {\it achievable} if there exists a sequence of $(2^{nR_1}, 2^{nR_2}, \cdots, 2^{nR_M}, n)$ codes with
\begin{align}
   \limsup_{n\rightarrow\infty} {\mathbb E}\left[ \frac{1}{n} \sum_{i=1}^n d(X_i, \hat{X}_i) \right] & \le D,
\end{align}
where $d(\cdot, \cdot)$ is a {\it (single-letter) distortion measure} (see~\cite[p.~304]{CoverThomas06}).
In much of the present paper, we restrict attention to the case of the mean-squared error distortion measure, {\it i.e.,}
\begin{align}
  d(x_i, \hat{x}_i)  &= (x_i- \hat{x}_i)^2.
\end{align}
The rate-distortion region ${\cal R}_{CEO}(D)$ for the CEO problem is the closure of the set of all tuples $(R_1, R_2, \cdots, R_M)$ such that $(R_1, R_2, \cdots, R_M, D)$ is achievable.
In the present study, we are mostly interested in the minimum sum-rate, {\it i.e.,} the quantity defined as
\begin{align}
R_X^{CEO}(D) &= \min_{(R_1, R_2, \cdots, R_M) \in {\cal R}_{CEO}(D)}  \sum_{m=1}^M R_m. \label{eq-def-minsumrate}
\end{align}

\subsection{The special case $M=1$}\label{sec-def-remote}

In the special case $M=1,$ the CEO problem is referred to as the {\it remote} source coding problem.
This problem dates back to Dobrushin and Tsybakov~\cite{dobrushintsybakov-it62} as well as, for the case of additive noise, to Wolf and Ziv~\cite{wolfziv-it70}.
We will use the notation $R_{X}^R(D)$ in place of $R_X^{CEO}(D)$ in this case.
Here, it is well known that (see e.g.~\cite[Sec.~3.5]{Berger71})
\begin{align}
  R_{X}^R(D) &= \min_{f(\hat{x}|y) : {\mathbb E}[d(X, \hat{X})] \le D} I(Y; \hat{X}). \label{eq-def-Remote-RDF}
\end{align}

\subsection{The AWGN CEO Problem}

In much of the present study, we are concerned with the case where the source observation process is given by a {\it Gaussian}  broadcast channel. In that case, we have that
$Y_m = X + Z_m,$ for $m = 1, 2, \ldots, M,$ where $Z_m$ is distributed as a zero-mean Gaussian of variance $\sigma_m^2.$
We will refer to this as the {\it AWGN CEO problem,} illustrated in Figure~\ref{fig-CEO}.
Moreover, when the distortion measure of interest is the mean-squared error, we mimic standard terminology and refer to the {\it quadratic AWGN CEO problem.}

For the AWGN CEO problem, it will be convenient to use the following shorthand. For any subset ${\cal A} \subseteq \{1, 2, \ldots, M\},$
the sufficient statistic for $X$ given $\{ Y_i \}_{i \in {\cal A}}$ can be expressed as
\begin{align}
\YA &= \frac{1}{|{\cal A}|}\sum_{i \in {\cal A}} \frac{\sA}{\sigma_{Z_i}^2} Y_i
\\
&=  X+\ZA, \label{eq-def-YA}
\end{align}
where 
\begin{align}
\ZA & = \frac{1}{|{\cal A}|}\sum_{i \in {\cal A}} \frac{\sA}{\sigma_{Z_i}^2} Z_i \label{eq-def-ZA}
\end{align}
is a zero-mean Gaussian random variable of variance $\sA/|{\cal A}|,$
and $\sA$ denotes the harmonic mean of the noise variances in the set ${\cal A},$ that is,
\begin{align}
\sA & = \left(\frac{1}{|{\cal A}|}\sum_{i \in {\cal A}}\frac{1}{\sigma_{Z_i}^2} \right)^{-1}. \label{eq-def-sA}
\end{align}
In the special case where ${\cal A} =  \{ 1, 2, \ldots, M\},$ we will use the notation
\begin{align}
\YM &= \frac{1}{M}\sum_{i=1}^M \frac{\sigma_{\cal M}^2}{\sigma_{Z_i}^2} Y_i
\\
&=  X+\ZM, \label{eq-def-YM}
\end{align}
where $\sigma_{\cal M}^2$ denotes the harmonic mean of all the noise variances and
\begin{align}
\ZM & = \frac{1}{M}\sum_{i=1}^M \frac{\sigma_{\cal M}^2}{\sigma_{Z_i}^2} Z_i, \label{eq-def-ZM}
\end{align}
respectively. Hence, $\ZM$ is a zero-mean Gaussian random variable of variance $\sigma_{\cal M}^2/M.$

\section{The Shannon Lower Bound and Its Extensions}\label{sec-known}

The Shannon lower bound concerns the rate-distortion function $R_X(D)$ for an arbitrary (not necessarily Gaussian) source $X$ subject to mean-squared error distortion. It states that
\begin{align}
R_{X}(D) & \ge \frac{1}{2} \log^+   \frac{\entp{X}}{D}.\label{eq:direct_lower_bd}
\end{align}
At the same time, a maximum entropy argument provides an upper bound to the same rate-distortion function:
\begin{align}
R_{X}(D) & \leq \frac{1}{2} \log^+ \frac{\sigma_X^2}{D}. \label{eq:direct_upper_bd}
\end{align}
These results date back to~\cite{shannon-ire59} (see also~{\cite[Eqns.~(4.3.32) and~(4.3.42)]{Berger71} or~\cite[p.~338]{CoverThomas06}}).
Part of their appeal is the interesting duality played by the source power and its entropy power.
This also directly implies their tightness in the case where the underlying source $X$ is Gaussian, since power and entropy-power are equal in that case.
As a side note, tangential to the discussion presented here, we point out that the (generalized) Shannon lower bound is not generally tight for Gaussian {\it vector} sources, see e.g.~\cite{GyorgyLZ:it99}.

One can extend this result rather directly to the case of the remote rate-distortion function, {\it i.e.,} the CEO problem with $M=1,$ as defined above in Section~\ref{sec-def-remote}.
Specifically, letting $V={\mathbb E}[X|Y],$ the remote rate-distortion function subject to mean-squared error satisfies the bounds (see Appendix~\ref{App-proofs-known})
\begin{align}
 \frac{1}{2} \log^+ \frac{\entp{V}}{D-D_0} \le  R_X^R(D)  \le \frac{1}{2} \log^+ \frac{\sigma_V^2}{D-D_0}, \label{eqn-remote-SLB}
\end{align}
for $D > D_0,$ where $D_0 = {\mathbb E}\left[ \left(X - V\right)^2  \right].$
For the special case of additive source observation noise, that is, $Y = X+Z,$ where $X$ and $Z$ are independent, one can obtain a more explicit pair of bounds by observing that (see Appendix~\ref{App-proofs-known})
\begin{align}
 \frac{\entp{X}\entp{Z}}{\entp{Y}} \le D_0 \le \frac{\sigma_X^2\sigma_Z^2}{\sigma_Y^2}. \label{eqn-estimation}
\end{align}
Combining Inequalities~\eqref{eqn-remote-SLB} and~\eqref{eqn-estimation}, we obtain the slightly weakened lower bound, for $D > \entp{X}\entp{Z}/\entp{Y},$
 \begin{align}
R_X^R(D) & \ge \frac{1}{2} \log^+ \frac{\entp{V}}{D}  + \frac{1}{2} \log^+  \frac{ \entp{Y}}{\entp{Y}-\frac{\entp{X}}{D}\entp{Z} } \label{eqn-remote-SLB-explicit-lower}
\end{align}
and the upper bound, for  $D > \sigma_X^2\sigma_Z^2/\sigma_{Y}^2,$
 \begin{align}
 R_X^R(D) &\le \frac{1}{2}\log^+\frac{\sigma_V^2}{D}  + \frac{1}{2}\log^+ \frac{\sigma_Y^2}{\sigma_Y^2 - \frac{\sigma_X^2}{D}{\sigma_{{Z}}^2}}. \label{eqn-remote-SLB-explicit-upper} \end{align}

A second type of lower bounds of a similar flavor can be derived from entropy power inequalities (EPI).
For these bounds to work, we restrict attention to the case of the {\it AWGN CEO problem} as defined above, {\it i.e.,}  the scenario where the underlying source $X$ is observed under independent zero-mean {\it Gaussian}  noise $Z$ of variance $\sigma_Z^2.$
Again, we let $Y=X+Z$ be the noisy source observation. Moreover, let us consider an {\it arbitrary} distortion measure,
and let $R_X(D)$ denote the (regular) rate-distortion function of the source $X$ subject to that distortion measure.
Then, a lower bound to the remote rate-distortion function subject to that arbitrary distortion measure is (see Appendix~\ref{App-proofs-known})
\begin{align}
R_{X}^R(D) \geq R_{X}(D)+ \frac{1}{2} \log^+ \frac{\entp{X}}{\entp{Y}-\sigma_{Z}^2 e^{2 R_{X}(D)}} , \label{eq-EPI-general}
\end{align}
for $D$ satisfying $\sigma_Z^2 e^{2 R_X(D)} < \entp{Y}.$

Moreover, if the following inequality can be satisfied
\begin{equation}
 \min_{g} {\mathbb E}[ d\left(X,g(X+Z+W)\right)] \leq D,  \label{eq:ineq_snr_so}
\end{equation}
where $W$ is an independent zero-mean Gaussian random variable with variance $\sigma_X^2/(e^{2r}-1)-\sigma_Z^2$, and the minimum is over all real-valued, measurable functions $g(\cdot)$, then for $0 \leq r \leq \frac{1}{2} \log (1+\sigma_X^2/\sigma_Z^2)$, an upper bound is (see Appendix~\ref{App-proofs-known})
\begin{align}
R_X^R&(D) \leq r + \frac{1}{2}\log^+ \frac{\sigma_X^2}{\sigma_Y^2 -
\sigma_Z^2 e^{2r}}  . \label{eq:remote_rd_upperbd_thm_so}
\end{align}
When we restrict attention to the case of mean-squared error distortion, we can obtain the following more explicit form for the lower bound, for $D > \entp{X}\sigma_Z^2/\entp{Y},$
\begin{align}
R_X^R(D)&\ge \frac{1}{2} \log^+ \frac{\entp{X}}{D}  + \frac{1}{2} \log^+  \frac{ \entp{X}}{\entp{Y}-\frac{\entp{X}}{D}\sigma_{Z}^2 }, \label{eq-EPI-explicit-lower}
\end{align}
and for the upper bound, for $D > \sigma_X^2\sigma_Z^2/\sigma_{{Y}}^2,$
\begin{align}
  R_X^R(D) &\le \frac{1}{2}\log^+\frac{\sigma_X^2}{D}  + \frac{1}{2}\log^+ \frac{\sigma_X^2}{\sigma_Y^2 - \frac{\sigma_X^2}{D}{\sigma_{{Z}}^2}}, \label{eq-EPI-explicit-upper}
\end{align}
Proofs of Inequalities~\eqref{eq-EPI-explicit-lower}-\eqref{eq-EPI-explicit-upper} are provided in Appendix~\ref{App-proofs-known}.
It is tempting to compare the lower bounds in Inequalities~\eqref{eqn-remote-SLB-explicit-lower} and~\eqref{eq-EPI-explicit-lower}, but there does not appear to be a simple relationship.

\section{Main Results}\label{sec-main}

\subsection{General Lower Bound}

Our main result is the following lower bound:

\begin{thm}\label{thm-CEO-aux}
For the $M$-agent AWGN CEO problem with an arbitrary continuous underlying source $X,$ constrained to having finite differential entropy, subject to an arbitrary distortion measure $d(\cdot,\cdot),$
if a rate-distortion tuple $(R_1, R_2, \cdots, R_M, D)$ is achievable, {\it i.e.,} if it satisfies
$(R_1, R_2, \ldots, R_M) \in {\cal R}_{CEO}(D),$
then there must exist non-negative real numbers $\{r_1, r_2, \ldots, r_M\}$ such that
for every (strict) subset ${\cal A} \subset \{ 1, 2, \ldots, M \},$
we have
\begin{align}
\sum_{i \in {\cal A}}R_i &\geq R_X(D) \nonumber \\
 &    - \frac{1}{2} \log \left( \frac{\entp{\YAc}}{\sAc/|{\cal A}^c|}- {\entp{X}}{}\sum_{i \in {\cal A}^c} \frac{e^{-2r_i}}{\sigma_{Z_i}^2} \right) +  \sum_{i \in {\cal A}}  r_i,\label{eq-thm-CEO-aux}
\end{align}
and for the full set ${\cal A} = \{ 1, 2, \ldots, M \},$ we have $\sum_{i \in {\cal A}}R_i \geq   R_X(D)  +  \sum_{i \in {\cal A}}  r_i,$
where $\YA$ and $\sA$ are defined in Equations~\eqref{eq-def-YA} and~\eqref{eq-def-sA}, respectively, and $R_X(D)$ denotes the (regular) rate-distortion function of the source $X$ with respect to the distortion measure $d(\cdot,\cdot).$
\end{thm}

The proof of this theorem is given in Appendix~\ref{App-proof-main}.

\begin{remark}
Note that the argument inside the logarithm in Equation~\eqref{eq-thm-CEO-aux} is lower bounded by $1$ for all non-negative choices of $r_i,$ as explained in the proof in Appendix~\ref{App-proof-main}, making the expression well-defined.
\end{remark}

In the next corollary, we specialize Theorem~\ref{thm-CEO-aux} to the case of the mean-squared error distortion measure, a case for which we have a closely matching upper bound.

\begin{cor}\label{Cor-CEO-aux-MMSE}
For the $M$-agent AWGN CEO problem with an arbitrary continuous underlying source $X,$ constrained to having finite differential entropy, subject to the mean-squared error distortion measure,
if a rate-distortion tuple $(R_1, R_2, \cdots, R_M, D)$ is achievable, {\it i.e.,} if it satisfies
$(R_1, R_2, \ldots, R_M) \in {\cal R}_{CEO}(D),$
then there must exist non-negative real numbers $\{r_1, r_2, \ldots, r_M\}$ such that
for every (strict) subset ${\cal A} \subset \{ 1, 2, \ldots, M \},$
we have
\begin{align}
\sum_{i \in {\cal A}}R_i &\geq   \frac{1}{2} \log^+ \frac{\entp{X}}{D}  \nonumber \\
& \,\,\, - \frac{1}{2} \log \left( \frac{\entp{\YAc}}{\sAc/|{\cal A}^c|}- {\entp{X}}{}\sum_{i \in {\cal A}^c} \frac{e^{-2r_i}}{\sigma_{Z_i}^2} \right) +  \sum_{i \in {\cal A}}  r_i,\label{eq-thm-CEO-aux-MMSE}
\end{align}
and for the full set ${\cal A} = \{ 1, 2, \ldots, M \},$ we have $\sum_{i \in {\cal A}}R_i \geq   \frac{1}{2} \log^+ \frac{\entp{X}}{D}  +  \sum_{i \in {\cal A}}  r_i,$
where $\YA$ and $\sA$ are defined in Equations~\eqref{eq-def-YA} and~\eqref{eq-def-sA}, respectively.

For achievability, if there exist non-negative real numbers $\{r_1, r_2, \ldots, r_M\}$ such that
for every (strict) subset ${\cal A} \subset \{ 1, 2, \ldots, M \},$ we have
\begin{align}
\sum_{i \in {\cal A}}R_i &\leq   \frac{1}{2} \log^+ \frac{\sigma_X^2}{D}  \nonumber \\
&  \,\,\, - \frac{1}{2} \log \left( \frac{\sigma_{\YAc}^2}{\sAc/|{\cal A}^c|}- {\sigma_{X}^2}{}\sum_{i \in {\cal A}^c} \frac{e^{-2r_i}}{\sigma_{Z_i}^2} \right) +  \sum_{i \in {\cal A}}  r_i,\label{eq-thm-CEO-ach-Oohama}
\end{align}
and for the full set ${\cal A} = \{ 1, 2, \ldots, M \},$ we have $\sum_{i \in {\cal A}}R_i \leq   \frac{1}{2} \log^+ \frac{\sigma_X^2}{D}  +  \sum_{i \in {\cal A}}  r_i,$
then we have that $(R_1, R_2, \ldots, R_M) \in {\cal R}_{CEO}(D).$
\end{cor}

For the proof of this corollary, we note that Inequality~\eqref{eq-thm-CEO-aux-MMSE} follows directly by combining Theorem~\ref{thm-CEO-aux} with the Shannon lower bound, Inequality~\eqref{eq:direct_lower_bd}.
The proof of the achievability part, Inequality~\eqref{eq-thm-CEO-ach-Oohama}, follows from the work of Oohama~\cite{oohama-it98,oohama-it05}. We briefly comment on this in Appendix~\ref{proof:thm-AWGN-CEO-upper}.

Comparing Inequalities~\eqref{eq-thm-CEO-aux-MMSE} and~\eqref{eq-thm-CEO-ach-Oohama}, 
we observe a pleasing duality of the source power and its entropy power: to go from the lower bound to the upper bound, it suffices to replace all entropy powers by the corresponding power (variance) of the same random variable.
This fact directly implies tightness for the case where the underlying source is Gaussian, which of course is well known~\cite{oohama-it05}.
The bounds also imply that for fixed source entropy power, the Gaussian is a best-case source,
and for fixed source power (variance), it is a worst-case source.

The same kind of duality can be observed in the Shannon lower bound in Inequalities~\eqref{eq:direct_lower_bd}-\eqref{eq:direct_upper_bd}.
It also appears in the extensions given in Inequality~\eqref{eqn-remote-SLB}, in Inequalities~\eqref{eqn-remote-SLB-explicit-lower}-\eqref{eqn-remote-SLB-explicit-upper}, and again in Inequalities~\eqref{eq-EPI-explicit-lower}-\eqref{eq-EPI-explicit-upper}.

\subsection{Sum-rate Lower Bound For Equal Noise Variances}

From Theorem~\ref{thm-CEO-aux}, we can obtain the following more explicit bound on the sum rate in the case when all observation noise variances are equal:

\begin{thm}\label{thm-AWGN-CEO-generallowerbound}
For the $M$-agent AWGN CEO problem with an arbitrary continuous underlying source $X,$ constrained to having finite differential entropy, with observation noise variance $\sigma_{Z_m}^2=\sigma_Z^2,$ for $m=1, 2, \ldots, M,$ and subject to an arbitrary distortion measure $d(\cdot,\cdot),$ the sum-rate distortion function is lower bounded by
\begin{align}
R_X^{CEO}(D) &\geq R_X(D) + \frac{M}{2} \log^+ \frac{M\entp{X}}{M\entp{\YM}  -  \sigma_{Z}^2 e^{2 R_X(D)} },
\label{eq:sum_rate_lb}
\end{align}
for $D$ satisfying $\sigma_Z^2 e^{2 R_X(D)} < M\entp{\YM},$ where $Y(M)$ is defined in Equation~\eqref{eq-def-YM},
and $R_X(D)$ denotes the (regular) rate-distortion function of the source $X$ with respect to the distortion measure $d(\cdot,\cdot).$
\end{thm}
The proof of this theorem is given in Appendix~\ref{proof:thm-AWGN-CEO-lower}.

When we further specialize to the case of the mean-squared error distortion measure, then our lower bound takes the same shape as a well-known achievable coding strategy, except that the lower bound has entropy powers where the upper bound has powers (variances). Specifically, we have the following result:

\begin{cor}\label{thm-AWGN-CEO}
For the $M$-agent AWGN CEO problem with an arbitrary continuous underlying source $X,$ constrained to having finite differential entropy, with observation noise variance $\sigma_{Z_m}^2=\sigma_Z^2,$ for $m=1, 2, \ldots, M,$ and subject to mean-squared error distortion, the CEO sum-rate distortion function is lower bounded by
\begin{align}
\lefteqn{R_X^{CEO}(D) \geq R_{X, lower}^{CEO}(D)} \nonumber \\
 & = \frac{1}{2} \log^+ \frac{\entp{X}}{D}  + \frac{M}{2} \log^+ \frac{M \entp{X}}{M\entp{\YM}-\frac{\entp{X}}{D}\sigma_{Z}^2 } \label{eq:ceo_rd_lowerbd_sqer}
\end{align}
for $D > \entp{X}\sigma_{Z}^2/(M\entp{\YM}).$
Moreover, in this case, the CEO sum-rate distortion function is upper bounded by
\begin{align}
\lefteqn{R_{X}^{CEO}(D)  \leq R_{X,upper}^{CEO}(D)} \nonumber \\
 & = \frac{1}{2}\log^+\frac{\sigma_X^2}{D} + \frac{M}{2}\log^+ \frac{M\sigma_X^2}{M\sigma_{\YM}^2 - \frac{\sigma_X^2}{D}{\sigma_{{Z}}^2}} ,
\label{eq:ceo_rd_upperbd_sqer}
\end{align}
for $D > \sigma_X^2\sigma_Z^2/(M\sigma_{\YM}^2),$
where $Y(M)$ is defined in Equation~\eqref{eq-def-YM}.
\end{cor}

The proof of this corollary is given in Appendices~\ref{proof:thm-AWGN-CEO-lower} and~\ref{proof:thm-AWGN-CEO-upper}.

\begin{remark}
We point out that $\sigma_{\YM}^2 = \sigma_X^2 + \sigma_Z^2/M,$ but we prefer to leave it in the shape given in the above corollary in order to emphasize the duality of the upper and the lower bound.
\end{remark}

To illustrate the power of the presented bounds in a formal way, we will restrict attention to the class of source distributions $f_X(x)$ for which $\kappa_X < \infty,$ where
\begin{align}
 \kappa_X &= \left.\frac{d}{ds}\left(\entp{X+\sqrt{s}G}\right)\right|_{s=0}, \label{Eq-def-kappaX}
\end{align} 
where $G$ is a zero-mean unit-variance Gaussian random variable, independent of $X.$
Note that in the special case where $X$ itself is Gaussian, we have $\kappa_X =1.$
For starters, let us suppose that the distortion $D$ is a constant, independent of $M.$
In this case, it can be verified that both Equations~\eqref{eq:ceo_rd_lowerbd_sqer} and~\eqref{eq:ceo_rd_upperbd_sqer} tend to constants as $M$ becomes large,
and we have (see Appendix~\ref{App-Eq-gap-upperlower})
\begin{align}
\lefteqn{\lim_{M\rightarrow \infty} \left( R_{X,upper}^{CEO}(D) - R_{X, lower}^{CEO}(D) \right) }\nonumber \\
& \le \frac{1}{2}\log\frac{\sigma_X^2}{\entp{X}} +  \frac{(\kappa_X\sigma_X^2-\entp{X})\sigma_Z^2}{2\sigma_X^2\entp{X}}. \label{Eq-limit-upperlower-constantD}
\end{align}
Note that the right-hand side can also be expressed as $D(f_X\|g_X) + \frac{\sigma_Z^2}{2\sigma_X^2} ( \kappa_X \exp(2D(f_X\|g_X))-1),$ where $g_X$ is a Gaussian probability density function with the same mean and variance as $f_X,$ and $D(\cdot \| \cdot)$ denotes the Kullback-Leibler divergence.
This illustrates how the gap between the upper and the lower bound narrows as $f_X$ gets closer to a Gaussian distribution.
Arguably a more interesting regime in the CEO problem is when the distortion $D$ {\it decreases} as a function of $M:$ the more observations we have, the lower a distortion we should ask for. A natural scaling is to require the distortion to decay inversely proportional to $M.$
Specifically, let us consider a distortion $D = d/M,$ where $d>\sigma_Z^2$ is a constant independent of $M.$
Then, it is immediately clear that both the upper and the lower bound in Corollary~\ref{thm-AWGN-CEO} increase {\it linearly}  with $M.$
But how does their gap behave with $M$?
This is a slightly more subtle question.
We can show that for all $M>d/\sigma_X^2,$ the difference between Equations~\eqref{eq:ceo_rd_lowerbd_sqer} and~\eqref{eq:ceo_rd_upperbd_sqer} is upper bounded by (see Appendix~\ref{App-Eq-gap-upperlower})
\begin{align}
\lefteqn{R_{X,upper}^{CEO}(d/M) - R_{X, lower}^{CEO}(d/M)} \nonumber \\
& \le \frac{1}{2}\log\frac{\sigma_X^2}{\entp{X}} +  \frac{(\kappa_X\sigma_X^2-\entp{X})\sigma_Z^2}{2\sigma_X^2\entp{X}(1-\frac{\sigma_Z^2}{d})}, \label{Eq-gap-upperlower}
\end{align}
which does not depend on $M.$ 
Hence, when interpreted as a function of the number of agents $M,$ the bounds of Corollary~\ref{thm-AWGN-CEO} capture the behavior rather tightly.

\subsection{Rate Loss for the quadratic AWGN CEO problem}\label{Sec-implications-sub-pricedistributed}

In this section, we restrict attention to the case of the quadratic AWGN CEO problem.
The rate loss is the difference in coding rate needed in the distributed coding scenario of Figure~\ref{fig-CEO} and the coding rate that would be required if the encoders could fully cooperate.
If the encoders fully cooperate, the resulting problem is precisely a remote rate-distortion problem as defined in Section~\ref{sec-def-remote}, where the source is observed in zero-mean Gaussian noise of variance $\sigma_Z^2/M.$ This follows directly from the observation that $\YM$ as defined in Equation~\eqref{eq-def-YM} is a sufficient statistic for the underlying source $X,$ given all the noisy observations. As before, we denote the remote rate-distortion function by $R_{X}^{R}(D),$ and hence, the rate loss is the difference $R_{X}^{CEO}(D)-R_{X}^{R}(D).$
It is known that the rate loss is {\it maximal}  when the underlying source $X$ is Gaussian~\cite[Proposition~4.3]{EswaranG:09}.
For example, in the case where the distortion $D$ is required to decrease inversely proportional to $M,$
the rate loss increases {\it linearly} as a function of the number of agents $M,$ and is thus very substantial. If $X$ is not Gaussian, may we end up with a much more benign rate loss?
Restricting again to sources of non-zero entropy power and for which $\kappa_X < \infty$ (see the definition given in Equation~\eqref{Eq-def-kappaX}), we can show that 
the answer to this question is {\it no.} This follows directly from the bounds established in this paper. Specifically, we have the following statement:

\begin{thm}\label{lemma-priceofdistributed}
For the $M$-agent AWGN CEO problem with an arbitrary continuous underlying source $X,$ constrained to having finite differential entropy and $\kappa_X < \infty,$ with observation noise variance $\sigma_{Z_m}^2=\sigma_Z^2,$ for $m=1, 2, \ldots, M,$ and subject to mean-squared error distortion,
letting the distortion $D_\alpha$ be parameterized as
\begin{align}
D_\alpha &= \alpha \frac{\sigma_X^2\sigma_Z^2}{M\sigma_X^2+\sigma_Z^2},
\end{align}
where $\alpha$ satisfies
\begin{align}
 1 < \alpha & \le \min \left\{M \frac{\sigma_X^2}{\sigma_Z^2}+1, M \frac{\entp{X} \entp{\YM} }{\sigma_Z^2 \sigma_X^2} \right\} \label{Eq-rateloss-condalpha}
\end{align}
and where $\YM$ is defined in Equation~\eqref{eq-def-YM},
the rate loss of distributed coding versus centralized coding is at least
\begin{align}
\lefteqn{R_{X}^{CEO}(D_\alpha)-R_{X}^{R}(D_\alpha)} \nonumber \\
& \ge  \frac{M}{2} \log^+  \frac{ \alpha \gamma_X }{ ( \alpha \gamma_X -1)    \left(1+ \frac{\kappa_X \sigma_Z^2}{M \entp{X}} \right) }   - \frac{1}{2} \log  \frac{\gamma_X^2\alpha}{\alpha-1},
\end{align}
where $\gamma_X = \sigma_X^2/\entp{X},$ where we note that $\gamma_X \ge 1.$
\end{thm}

The proof is given in Appendix~\ref{app-proof-Sec-implications}. Note that the rate loss has to be non-negative, hence our formula can be slightly improved by only keeping the positive part. We prefer not to clutter our notation with this since it becomes immaterial as soon as $M$ gets large.

While the bound of Theorem~\ref{lemma-priceofdistributed} is valid for all choices of the parameters, it is arguably most interesting when interpreted as a function of the number of agents $M.$ When $\alpha$ is a constant independent of $M$ and thus, the distortion decreases inversely proportional to $M,$ it is immediately clear that the rate loss increases linearly with $M.$

\section{Joint Source-Channel Coding}\label{Sec-implications-sub-digital}

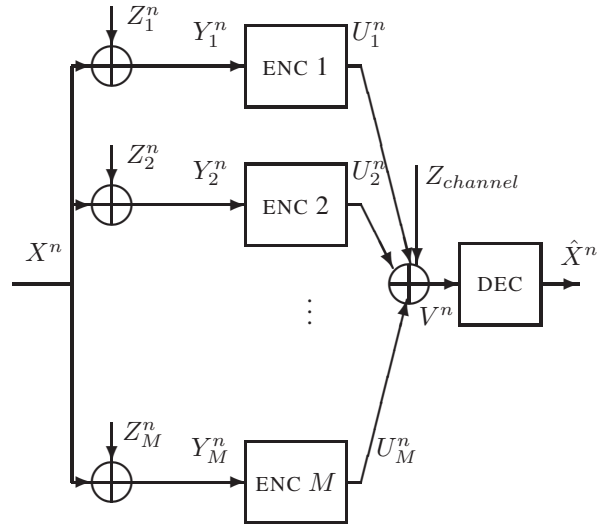
\begin{figure}
  \begin{center}
  \setlength{\unitlength}{1.5pt}
  \begin{picture}(140,130)(130,0)
\thicklines
    \put (130, 60) {\line (1, 0) {15} }
    \put (133, 63) {\makebox (10,10) {$X^n$}}

    \put (145, 10) {\line (0, 1) {105} }

    \put (145, 115) {\vector (1, 0) {5} }
    \put (150, 115) {\line (1, 0) {20} }
    \put (155, 115) {\circle {10} }
    \put (155, 110) {\line (0, 1) {10} }
    \put (155, 130) {\vector (0, -1) {10} }
    \put (158, 122) {\makebox (10,10) {$Z_1^n$}}

    \put (145, 80) {\vector (1, 0) {5} }
    \put (150, 80) {\line (1, 0) {20} }
    \put (155, 80) {\circle {10} }
    \put (155, 75) {\line (0, 1) {10} }
    \put (155, 95) {\vector (0, -1) {10} }
    \put (158, 87) {\makebox (10,10) {$Z_2^n$}}

    \put (145, 10) {\vector (1, 0) {5} }
    \put (150, 10) {\line (1, 0) {20} }
    \put (155, 10) {\circle {10} }
    \put (155, 5) {\line (0, 1) {10} }
    \put (155, 25) {\vector (0, -1) {10} }
    \put (158, 17) {\makebox (10,10) {$Z_M^n$}}

    \put (170, 115) {\vector (1, 0) {19} }
    \put (175, 118) {\makebox (10,10) {$Y_1^n$}}
    \put (189, 105) {\framebox (25,20) {\sc enc 1} }
    \put (214, 115) {\line (1, 0) {4} }
    \put (218, 115) {\vector (1, -4) {12.5} }
    \put (215,118) {\makebox (10,10) {$U_1^n$}}

    \put (170, 80) {\vector (1, 0) {19} }
    \put (175, 83) {\makebox (10,10) {$Y_2^n$}}
    \put (189, 70) {\framebox (25,20) {\sc enc 2} }
    \put (214, 80) {\line (1, 0) {4} }
    \put (218, 80) {\vector (1, -2) {8} }
    \put (215, 83) {\makebox (10,10) {$U_2^n$}}

    \put (200, 50) {\makebox (10,10) {$\vdots$}}

    \put (170, 10) {\vector (1, 0) {19} }
    \put (175, 13) {\makebox (10,10) {$Y_M^n$}}
     \put (189, 0) {\framebox (25,20) {\sc enc $M$} }
    \put (214, 10) {\line (1, 0) {4} }
    \put (218, 10) {\vector (1, 4) {11.5} }
    \put (222,13) {\makebox (10,10) {$U_M^n$}}

    \put (232, 90) {\vector (0, -1) {25} }
    \put (241,82) {\makebox (10,10) {$Z_{channel}$}}

    \put (230, 60) {\circle {10} }
    \put (225, 60) {\line (1, 0) {10} }
    \put (230, 55) {\line (0, 1) {10} }

    \put (235, 60) {\vector (1, 0) {8} }
    \put (232, 47) {\makebox (10,10) {$V^n$}}

    \put (243, 50) {\framebox (20,20) {\sc dec} }
    \put (263, 60) {\vector (1, 0) {10} }
    \put (268, 63) {\makebox (10,10) {${\hat X}^n$}}
  \end{picture}
  \end{center}
\caption{A network joint source--channel coding problem inspired by the CEO problem. $X$ is an arbitrary source with variance (power) $\sigma_X^2$ (not necessarily Gaussian) and entropy power $\entp{X}.$ Each encoder can produce a codeword $U_i^n$ of average power no more than $P,$ which is then transmitted over a standard symmetric additive white Gaussian noise multiple-access channel.}
\label{fig-JSCC}
\end{figure}

One important application of the new bound presented here is to network joint source-channel coding.

\subsection{Problem Statement}

The ``sensor'' network considered in this section is illustrated in Figure~\ref{fig-JSCC}.
The underlying source $X$ and the source observation process are exactly as in the AWGN CEO problem defined above,
and we will only consider the simple symmetric case where all observation noise variances are equal, that is, $\sigma_{Z_m}^2=\sigma_Z^2,$ for $m=1, 2, \ldots, M.$
Additionally, in the present section, we restrict attention to those source distributions $f_X(x)$ for which $\kappa_X<\infty,$
where $\kappa_X$ is as defined in Equation~\eqref{Eq-def-kappaX}.

With reference to Figure~\ref{fig-JSCC}, encoder $m$ can apply an {\em arbitrary} sequence of real-valued coding functions $f_{m, i}(\cdot),$
for $i=1, 2, \ldots, n,$ to the observation sequence
such as to generate a sequence of channel inputs,
\begin{align}
U_m[i] &= f_{m, i}(Y_m[1], Y_m[2], \ldots, Y_m[n]).
\end{align}
The only constraint is that the functions $f_{m, i}(\cdot)$ be chosen to ensure that
\begin{eqnarray}
   \frac{1}{n} \sum_{i=1}^n {\mathbb E}\left[ ( U_m[i] )^2 \right] & \le & P,  \label{Eq-PowerCon}
\end{eqnarray}
for $m=1, 2, \ldots, M.$
For $i=1, 2, \ldots, n,$
the channel outputs are given by
\begin{eqnarray}
  V[i] & = & Z_{channel}[i] + \sum_{m=1}^M U_m[i],
\end{eqnarray}
where
$\{ Z_{channel}[i] \}_{i=1}^n$ is an i.i.d. sequence of Gaussian random variables
of mean zero and variance $\sigma_{channel}^2.$
Upon observing the channel output sequence $\{ V[i] \}_{i=1}^n,$
the decoder (or fusion center) must produce a sequence
$\hat{X}[i] = g_i (V[1], V[2], \ldots, V[n]).$
A power-distortion pair $(P, D)$ is said to be {\it achievable}  if there exists a sequence of sets of mappings
$\{ f_{m, i}(\cdot) \}_{i=1}^n,$ for $m=1, 2, \ldots, M,$ and $\{g_i\}_{i=1}^n$ (a sequence as a function of $n$) with
\begin{align}
   \limsup_{n\rightarrow\infty} \frac{1}{n} \sum_{i=1}^n {\mathbb E}\left[ ( X[i] - \hat{X}[i] )^2 \right] & \le D .\label{Eq-DistDef}
\end{align}
The power-distortion region for this network joint source-channel coding problem is the closure of the set of all achievable power-distortion pairs.

\subsection{Main Result}

The main result of this section is an assessment of the performance of {\it digital}  communication strategies for the communication problem illustrated in Figure~\ref{fig-JSCC}.
To put this in context, it is important to recall the so-called {\it source-channel separation theorem}  due to Shannon, see e.g.~\cite[Sec.~7.13]{CoverThomas06}.
For stationary ergodic point-to-point communication, this theorem establishes that it is without fundamental loss of optimality to compress the source to an index (that is, a bit stream) and then to communicate this index in a reliable fashion across the channel using capacity-approaching codes.
Such strategies are commonly known as {\it digital}  communication and are the underpinnings of most of the existing communication systems.

It is well-known that source-channel separation is suboptimal in {\it network}  communication settings, see e.g.~\cite[p.~592]{CoverThomas06}.
This suboptimality can be very substantial.
Specifically, for the example scenario as in Figure~\ref{fig-JSCC}, but where the underlying source $X$ is Gaussian, it was shown in~\cite[Sec.~5.4.6]{gastpar:02thesis}
that the suboptimality manifests itself as an {\it exponential}  gap in scaling behavior when viewed as a function of the number of nodes in the network.\footnote{In fact, for this special case, the optimal performance was characterized precisely in~\cite{Gastpar:08}.}
Could this gap be less dramatic for sources $X$ that are not Gaussian?
The new bounds established in the present paper allow to answer this question in the {\it negative.}
Specifically, we have the following result:

\begin{thm}\label{thm-jscc}
For the joint source-channel network considered in this section, if each encoder first compresses its noisy source observations into an index using the optimal CEO source code,
and this index is then communicated reliably over the multiple-access channel, the resulting power-distortion region must satisfy
\begin{align}
D_{\mathrm{d}} & \ge  \frac{\entp{X}\sigma_Z^2}{ \entp{X}\log \left( 1 + M^2\frac{ P}{\sigma^2_{channel}} \right) + \kappa_X \sigma_Z^2}.\label{Eq-scalinglaw-digital}
\end{align}
By contrast, there exists an (analog) communication strategy that incurs a distortion of 
\begin{align}
     D_{\mathrm{a}} & =   \frac{ \sigma_X^2 \sigma_Z^2}{M \sigma_X^2 + \sigma_Z^2 }
\left( 1 +  \frac{ M ( \sigma_X^2\sigma_{channel}^2/ \sigma_Z^2 )  }{\frac{M\sigma_X^2+\sigma_Z^2}{\sigma_X^2+\sigma_Z^2} MP +\sigma_{channel}^2} \right). \label{Eq-scalinglaw-analog}
\end{align}
\end{thm}

A proof of this theorem is given in Appendix~\ref{App-thm-jscc}.

The insight of Theorem~\ref{thm-jscc} lies in the comparison of Inequality~\eqref{Eq-scalinglaw-digital} with Equation~\eqref{Eq-scalinglaw-analog}.
Namely, the dependence of the attainable distortion on the number of agents $M:$ As one can see, for digital architectures, characterized by Inequality~\eqref{Eq-scalinglaw-digital},
the distortion decreases inversely proportional to the {\it logarithm}  of $M.$ By contrast, from Equation~\eqref{Eq-scalinglaw-analog}, there is a scheme for which the decrease is inversely proportional to $M.$ This represents an {\it exponential}  gap in the scaling-law behavior.
In other words, in order to attain a certain fixed desired distortion level $D,$ the number of agents needed in a digital architecture is {exponentially larger}  than the corresponding number for a simple analog scheme.
Hence, the bounds presented here imply that the exponential suboptimality of digital coding strategies observed in~\cite[Thm.~1 versus Thm.~2]{Gastpar:08} continues to hold for a large class of underlying sources $X$ with non-zero entropy power.

\section*{Acknowledgements}

The authors acknowledge helpful discussions with Aaron Wagner, Vinod Prabhakaran, Paolo Minero, Anand Sarwate, and Bobak Nazer, who were all part of the Wireless Foundations Center at the University of California, Berkeley, when this work was started.
They also thank Mich\`ele Wigger for comments on the manuscript.

\appendices

\section{Proofs for Section~\ref{sec-known}}\label{App-proofs-known}

\subsection{Proofs of Inequalities~\eqref{eqn-remote-SLB} and~\eqref{eqn-estimation}}

For Inequality~\eqref{eqn-remote-SLB}, we start by considering the (remote) distortion-rate function, $D_X^R(R),$ that is, the dual version of the minimization problem in Equation~\eqref{eq-def-Remote-RDF}, which can be expressed as
\begin{align}
  \lefteqn{D_X^R(R)} \nonumber \\
  & =  \min_{f(\hat{x}|y) : I(Y; \hat{X}) \le R}  {\mathbb E}[|X-\hat{X}|^2] \nonumber \\
    & =  \min_{f(\hat{x}|y) : I(Y; \hat{X}) \le R}  {\mathbb E}\left[\left|X- {\mathbb E}[X|Y] +{\mathbb E}[X|Y] -    \hat{X}\right|^2\right]  \nonumber \\
 & =   \min_{f(\hat{x}|y) : I(Y; \hat{X}) \le R}  {\mathbb E}\left[\left|X- {\mathbb E}[X|Y] \right|^2 \right]  +{\mathbb E}\left[\left|{\mathbb E}[X|Y] -    \hat{X}\right|^2\right]  \nonumber
 \end{align}
 by the properties of the conditional expectation. We thus obtain
 \begin{align}
  \lefteqn{D_X^R(R)} \nonumber \\
        & =     \underbrace{{\mathbb E}\left[\left|X- {\mathbb E}[X|Y] \right|^2 \right]}_{=D_0}  + \min_{f(\hat{x}|y) : I(Y; \hat{X}) \le R}  {\mathbb E}\left[\left|\underbrace{{\mathbb E}[X|Y]}_{=V} -    \hat{X}\right|^2\right]  \nonumber \\
     & =  D_0 + \underbrace{\min_{f(\hat{x}|v) : I(V; \hat{X}) \le R}{\mathbb E}\left[\left|V -    \hat{X}\right|^2\right]}_{D_V(R)} .
  \end{align}
   where for the last step, the data processing inequality implies that $I(Y; \hat{X}) \ge I(V; \hat{X}),$
 and hence, the second minimum cannot evaluate to something larger than the first.
Since $V$ is a deterministic function of $Y,$ we have that for the minimizing $f(\hat{x}|v)$ in the second minimum, it holds that $I(Y; \hat{X}) = I(V; \hat{X}).$
Hence, the two minima are equal.
 Conversely, we can thus write
  \begin{align}
  R_X^R(D) & =  R_V(D - D_0),
  \end{align}
where $R_V(D)$ denotes the rate-distortion function (under mean-squared error) of the source $V.$
The claimed lower and upper bounds now follow from Equations~\eqref{eq:direct_lower_bd} and~\eqref{eq:direct_upper_bd}, applied to $R_V(D)$.

For Inequality~\eqref{eqn-estimation}, the upper bound is simply the distortion incurred by the best linear estimator.
For the lower bound, observe that since by assumption, we can recover $X$ to within distortion $D_0$ from $Y,$ we must have
\begin{align}
I(X; Y) &\ge \min_{f(\hat{x}|x): {\mathbb E}[d(X, \hat{X})] \le D_0} I(X; \hat{X}) = R_X(D_0).
\end{align}
Under mean-squared error distortion, we know from Inequality~\eqref{eq:direct_lower_bd} that $R_X(D_0) \ge  \frac{1}{2} \log^+ \frac{\entp{X}}{D_0}.$
Combining this with the above, we obtain
\begin{align}
 I(X; Y) & \ge \frac{1}{2} \log^+ \frac{\entp{X}}{D_0}.
\end{align}
First, let us restrict to the case where $D_0 \le \entp{X}.$
In this case, we can further conclude that
\begin{align}
 e^{2I(X;Y)} & \ge \frac{\entp{X}}{D_0}.
\end{align}
Observing that $I(X;Y) = h(Y) - h(Z),$ we can rewrite this as
\begin{align}
 \frac{\entp{Y}}{\entp{Z}} & \ge \frac{\entp{X}}{D_0},
\end{align}
which is exactly the claimed bound.
Conversely, suppose that $D_0 > \entp{X}.$
By the entropy power inequality, we have that $\frac{\entp{Z}}{\entp{Y}} \le 1,$
meaning that the left-hand side of Inequality~\eqref{eqn-estimation} evaluates to something {\it no larger} than $\entp{X}.$
Since we assumed that  $D_0 > \entp{X},$ the claimed lower bound applies in this case, too.

\subsection{Proofs of Inequalities~\eqref{eq-EPI-general}-\eqref{eq-EPI-explicit-upper}}

\subsubsection*{Lower Bounds}
Recall that here, we are assuming that the observation noise $Z$ is Gaussian. Then, the lower bound in Inequality~\eqref{eq-EPI-general} can be established e.g. as a consequence of~\cite[Thm.1]{Courtade:it18}, as follows.
\begin{align}
  R^R_X(D) & = \min_{f(\hat{x}|y) : {\mathbb E}[d(X, \hat{X})] \le D} I(Y; \hat{X})   \\
      & \ge\min_{f(\hat{x}|y) : {\mathbb E}[d(X, \hat{X})] \le D} \frac{1}{2} \log \frac{e^{2h(X)}}{e^{2(h({Y})-I(X;\hat{X}))}-e^{2h(Z)}}.
\end{align}
where the inequality is due to~\cite[Thm.~1]{Courtade:it18} and the fact that by construction, we have that the Markov chain $X \leftrightarrow Y \leftrightarrow \hat{X}$ holds. 
Next, we observe that by definition, $R_X(D) \le I(X; \hat{X}).$ As long as $D$ is such that $e^{2 R_X(D)} \le \entp{Y}/\sigma_{Z}^2,$
the denominator stays non-negative. For such values of $D,$ we thus have
\begin{align}
  R^R_X(D)       & \ge\frac{1}{2} \log \frac{\entp{X}}{\entp{Y}e^{-2R_X(D)}-\sigma_{Z}^2}   \\
      & = R_X(D) + \frac{1}{2} \log \frac{\entp{X}}{\entp{Y} - \sigma_{Z}^2 e^{2 R_X(D)}}.
\end{align}
Finally, since for all values of $D,$ we have $R^R_X(D) \ge  R_X(D),$ we obtain
\begin{align}
  R^R_X(D)       & \ge  R_X(D) + \frac{1}{2} \log^+ \frac{\entp{X}}{\entp{Y} - \sigma_{Z}^2 e^{2 R_X(D)}}.
\end{align}
For the lower bound in Inequality~\eqref{eq-EPI-explicit-lower}, it suffices to lower bound $R_X(D)$ in Inequality~\eqref{eq-EPI-general} using Inequality~\eqref{eq:direct_lower_bd}.

\subsubsection*{Upper Bounds}
For the upper bound in Inequality~\eqref{eq:remote_rd_upperbd_thm_so}, let us consider $U = Y + V = X+Z+V,$ where $V$ is Gaussian $\mathcal{N}(0, \sigma_X^2/s-\sigma_Z^2).$
Now, let us suppose that $s$ can be chosen in such a way that
\begin{equation}
 \min_{g} {\mathbb E}[ d\left(X,g(U)\right)]  \leq D .
\end{equation}
Then, from the definition of the remote rate-distortion function (Equation~\eqref{eq-def-Remote-RDF}), we find
\begin{align}
R_X^R(D) &\leq I(Y;U) \\ &= h(U) - h(V) \\
&\leq \frac{1}{2} \log \left( 2\pi e \sigma_X^2\frac{1+s}{s}\right)-h(V) \\
&=\frac{1}{2} \log \left( \sigma_X^2\frac{1+s}{\sigma_X^2-s\sigma_Z^2 }\right)
\\
&= \frac{1}{2}\log \left(1+s\right)+ \frac{1}{2} \log \left(
\frac{\sigma_X^2}{\sigma_X^2-s\sigma_Z^2 }\right),
\end{align}
where the second inequality is a standard maximum-entropy argument.
To bring out the similarity to the corresponding lower bound, we reparameterize as $s = e^{2r}-1.$
For the upper bound in Inequality~\eqref{eq-EPI-explicit-upper}, we now observe that under mean-squared error distortion, as long as $D < \sigma_X^2,$ we may choose
\begin{align}
r & = \frac{1}{2} \log \frac{\sigma_X^2}{D},
\end{align}
or, equivalently, $\sigma_X^2/s = \frac{\sigma_X^2D}{\sigma_X^2-D}.$
To see that this is a valid choice satisfying the restriction of Equation~\eqref{eq:ineq_snr_so}, it suffices to observe that
\begin{align}
  \min_{\alpha \in {\mathbb R}} {\mathbb E}\left[ \left(X - \alpha(X+Z+V)\right)^2\right] & = D,
\end{align}
and thus we satisfy $\min_{g} {\mathbb E}[ d\left(X,g(X+Z+V)\right)]  \leq D.$
Finally, for $D \ge \sigma_X^2,$ the upper bound in Inequality~\eqref{eq-EPI-explicit-upper} evaluates to zero, which is trivially a correct bound, too.

\section{Proofs for Section~\ref{sec-main}}

\subsection{Proof of Theorem~\ref{thm-CEO-aux}}\label{App-proof-main}

The starting point for our lower bound is an outer bound introduced by Wagner and Anantharam
\cite{wagnerthesis-05,WagnerA:08}.
To state this bound, we write the vector of noisy observations as $\mathbf{Y} = (Y_1, Y_2, \ldots, Y_M)$
and we collect the elements $Y_i$ with $i$ in a subset ${\cal A}$ of the set $\{1, 2, \cdots, M\}$ into a vector
\begin{align}
\mathbf{Y}_{\cal A} &= ( Y_i )_{i\in{\cal A}},
\end{align}
and likewise, we introduce the auxiliary random vector $\mathbf{U} = (U_1, U_2, \ldots, U_M)$ and again collect the elements $U_i$ with $i$ in a subset ${\cal A}$ of the set $\{1, 2, \cdots, M\}$ into a vector
\begin{align}
\mathbf{U}_{\cal A} &= ( U_i )_{i\in{\cal A}}.
\end{align}
Then, the following statement applies.

\begin{thm}\label{thm-wagnerA}
Let $R_i$ denote the rate of the description provided by agent $i.$
There must exist a set of random variables $(X,\mathbf{Y},\mathbf{U},W,T,\hat{X}) \in \mathcal{W}_X^{CEO}(D)$
such that for all subsets ${\cal A} \subseteq \{ 1, 2, \ldots, M \},$
\begin{align} \sum_{i \in {\cal A}}R_i &\geq I(X;\mathbf{U},T) -
I(X; \mathbf{U}_{{\cal A}^c}|T) + \sum_{i \in {\cal A}} I(Y_i;U_i|X,W,T) ,
\label{eq:wagner_outer_bd} \end{align}
where
$\mathcal{W}_X^{CEO}(D)$ is the set of sets of random variables $(X,\mathbf{Y},\mathbf{U},W,T,\hat{X})$
satisfying ${\mathbb E}\left[ d(X,\hat{X}) \right]\leq D$ and
\begin{description}
    \item[$(i)$] $(W,T)$ is independent of $(X,\mathbf{Y})$,
    \item[$(ii)$] $\mathbf{U}_{\cal B} \leftrightarrow (\mathbf{Y}_{\cal B}, W, T) \leftrightarrow (X, \mathbf{Y}_{{\cal B}^c}, \mathbf{U}_{{\cal B}^c})$ for all ${\cal B} \subseteq \{1, \ldots M\}$,
    \item[$(iii)$] $(X,W,T) \leftrightarrow (\mathbf{U},T) \leftrightarrow \hat{X}$ , and
    \item[$(iv)$] the conditional distribution of $U_i$ given $W$ and $T$ is discrete for each $i$.
\end{description}
\end{thm}

For a proof of this theorem, see~\cite[p.~109]{wagnerthesis-05} or~\cite[Theorem~1, Appendix D, and start of the proof of Proposition 6]{WagnerA:08}.
Strictly speaking, in that proof, both the source and the observation noises are assumed to be Gaussian, but all arguments continue to hold for sources of finite differential entropy observed in Gaussian noise.

From this theorem, the following corollary will be of specific interest to our development:

\begin{cor}
There must exist a set of random variables $(X,\mathbf{Y},\mathbf{U},W,T,\hat{X}) \in \mathcal{W}_X^{CEO}(D)$
such that for all subsets ${\cal A} \subseteq \{ 1, 2, \ldots, M \},$
\begin{align}
\sum_{i \in {\cal A}}R_i &\geq  R_X(D)  - I(X; \mathbf{U}_{{\cal A}^c}|W,T) +  \sum_{i \in {\cal A}} I(Y_i;U_i|X,W,T) . \label{eq:wagner_sumrate_lb}
\end{align}
\end{cor}

\begin{proof}
Condition $(iii)$ in Theorem~\ref{thm-wagnerA} implies that $I(X;\mathbf{U},T) \geq I(X;\hat{X}) \geq R_X(D).$
Moreover, observe that
\begin{align}
\lefteqn{ I(X; \mathbf{U}_{{\cal A}^c}|T) + I(X; W | \mathbf{U}_{{\cal A}^c},T)} \nonumber \\
  & = I(X; W | T) +  I(X; \mathbf{U}_{{\cal A}^c}|W, T) 
\end{align}
and since $I(X; W | T)=0,$ we have $I(X; \mathbf{U}_{{\cal A}^c}|T) \le I(X; \mathbf{U}_{{\cal A}^c}|W,T).$
\end{proof}

To establish our lower bound, we start by considering
the following lemma. This is a generalization of the lemma proved by Oohama
\cite{oohama-it05} to the case of non-Gaussian sources.
\begin{lemma}
Let $r_i = I(Y_i;U_i|X,W,T)$ and ${\cal A} \subseteq \{1,\ldots,M\}$. Then
\begin{equation}
e^{2 I(X; \mathbf{U}_{\cal A}|W,T)} \leq \frac{\entp{\YA}}{\sA/|{\cal A}|} - \entp{X}\sum_{i \in {\cal A}} \frac{e^{-2r_i}}{\sigma_{Z_i}^2}.\label{generalized-oohama-bound}
\end{equation}

\begin{proof}
Since $(W,T)$ is independent of $(X, \mathbf{Y})$ when condition $(i)$ in Theorem~\ref{thm-wagnerA} holds,
we know that we preserve the Markov chain $X \rightarrow \YA \rightarrow \mathbf{Y}_{\cal A} \rightarrow \mathbf{U}_{\cal A}$ when we condition on any
realization of $(W,T).$
Therefore, we can again use Theorem 1 of~\cite{Courtade:it18} to infer
\begin{align}
\lefteqn{e^{2 h(\YA)}e^{-2 I(X; \mathbf{U}_{\cal A}|W=w,T=t)}} \nonumber \\
 &\geq e^{2 h(X)}e^{-2 I(\YA ; \mathbf{U}_{\cal A}|W=w,T=t)}+ e^{2 h(\ZA)} \nonumber \\
&= \frac{e^{2 h(X)}}{e^{2h(\YA)}}e^{2 h(\YA| \mathbf{U}_{\cal A},W=w,T=t)}+ e^{2 h(\ZA)}. \label{eq:oohama_gen1}
\end{align}
Now,
\begin{align}
h(\YA| &\mathbf{U}_{\cal A},W=w,T=t) \notag \\ = &h(\YA|\mathbf{U}_{\cal A}, X,W=w,T=t)\nonumber \\
 &+I(X;\YA|\mathbf{U}_{\cal A},W=w,T=t) \\
= &h(\YA| \mathbf{U}_{\cal A}, X,W=w,T=t)\nonumber \\
& + I(X;\YA,\mathbf{U}_{\cal A}|W=w,T=t) \notag \\
&-I(X;\mathbf{U}_{\cal A}|W=w,T=t) \\
= &h(\YA| \mathbf{U}_{\cal A}, X,W=w,T=t)\nonumber \\
& +
I(X;\YA|W=w,T=t) \notag \\
&- I(X;\mathbf{U}_{\cal A}|W=w,T=t) \label{eq:oohama_suff_stat}\\
= &h(\YA| \mathbf{U}_{\cal A}, X,W=w,T=t)+ I(X;\YA)\notag \\ 
&- I(X;\mathbf{U}_{\cal A}|W=w,T=t), \label{eq:oohama_gen2}
\end{align}
where~\eqref{eq:oohama_suff_stat} follows from the Markov chain $X \rightarrow \YA \rightarrow \mathbf{Y}_{\cal A} \rightarrow \mathbf{U}_{\cal A}$
and~\eqref{eq:oohama_gen2} from Theorem~\ref{thm-wagnerA}, Item ($i$). The next step is to bound
$h(\YA| \mathbf{U}_{\cal A}, X,W=w,T=t)$. We note that we can write
\begin{align}
\lefteqn{e^{2 h(\YA| \mathbf{U}_{\cal A}, X,W=w,T=t)}} \nonumber \\
 &= e^{2 h\left( \left. \frac{1}{|{\cal A}|} \sum_{i \in{\cal A}} \frac{\sA}{\sigma_{Z_i}^2} Y_i \right| \mathbf{U}_{\cal A}, X,W=w,T=t\right)} \\
&\geq \sum_{i \in {\cal A}} \left( \frac{\sA}{\sigma_{Z_i}^2}\right)^2 \frac{e^{2 h(Y_i| U_i, X,W=w,T=t)}}{|{\cal A}|^2}, \label{eq:oohama_gen3}
\end{align}
where (\ref{eq:oohama_gen3}) follows by Item ($ii$) of Theorem~\ref{thm-wagnerA} and the entropy power
inequality. Moreover,
\begin{align}
\lefteqn{h(Y_i| U_i, X,W=w,T=t) } \nonumber \\
&= h(Y_i|X,W=w,T=t)-I(Y_i;U_i|X,W=w,T=t) \nonumber \\
&= h(Z_i) - I(Y_i;U_i|X,W=w,T=t). \label{eq:oohama_gen4}
\end{align}
Combining~\eqref{eq:oohama_gen1},~\eqref{eq:oohama_gen2},~\eqref{eq:oohama_gen3}, and~\eqref{eq:oohama_gen4} gives
\begin{align}
\lefteqn{ e^{2 h(\YA)}e^{-2I(X; \mathbf{U}_{\cal A}|W=w,T=t)} } \nonumber \\
 \geq & \frac{e^{2 h(X)}}{e^{2h(\YA)}} \cdot \frac{e^{2I(X;\YA)}}{e^{2 I(X;\mathbf{U}_{\cal A}|W=w,T=t)}} \notag
\\
&\cdot \sum_{i \in {\cal A}} \left( \frac{\sA}{\sigma_{Z_i}^2}\right)^2 \frac{e^{2h(Z_i)-2I(Y_i;U_i|X,W=w,T=t)}}{|{\cal A}|^2} + e^{2 h(\ZA)} .
\end{align}
Solving for $e^{2I(X; \mathbf{U}_{\cal A}|W=w,T=t)}$ and noting that $e^{2h(Z_i)}/\sigma_{Z_i}^2  = 2\pi e,$ we get that
\begin{align}
\lefteqn{e^{2 I(X; \mathbf{U}_{\cal A}|W=w,T=t)}} \nonumber \\
 &\leq e^{-2 h(\ZA)} \nonumber \\
 & ~\cdot \left[ e^{2 h(\YA)}- \frac{e^{2h(X)}e^{2 I(X;\YA)}}{e^{2h(\YA)}}\right. \nonumber \\
& \,\,\,\,\,\, \,\,\,\,\,\,\,\,\,\,\,\,\,\,\,\,\,\,\,\,\,\,\,\,\,\,\,\,\,\,\,\,\,\, \left. 2 \pi e \left( \frac{\sA}{|{\cal A}|}\right)^2 \sum_{i \in {\cal A}}
\frac{e^{-2 I(Y_i;U_i|X,W=w,T=t)}}{ \sigma_{Z_i}^2}\right] \nonumber \\
&= \frac{e^{2 h(\YA)}}{e^{2 h(\ZA)}} \nonumber \\
&\,\,\,\,\,\,\,\,\,\,\,\,\,\,\,  - \frac{e^{2 h(X)}}{e^{4 h(\ZA)}}2 \pi e \left( \frac{\sA}{|{\cal A}|}\right)^2\sum_{i \in {\cal A}}\frac{e^{-2 I(Y_i;U_i|X,W=w,T=t)}}{\sigma_{Z_i}^2} \nonumber \\
& =  \frac{\entp{\YA}}{\sA/|{\cal A}|} - \entp{X}\sum_{i \in {\cal A}}\frac{e^{-2 I(Y_i;U_i|X,W=w,T=t)}}{\sigma_{Z_i}^2},
\end{align}
where we have used that $e^{2 h(\ZA)} = 2 \pi e \frac{\sA}{|{\cal A}|},$ a direct consequence of Equation~\eqref{eq-def-ZA}.
We complete the proof by taking the expectation over $W,T$ and applying
Jensen's inequality twice, once to the left-hand side and once to the
right-hand side.
\end{proof}
\end{lemma}

\begin{proof}[Proof of Theorem~\ref{thm-CEO-aux}]
Let $r_i = I(Y_i;U_i|X,W,T).$
First, if ${\cal A} = \{ 1, 2, \ldots, M \},$ then Inequality~\eqref{eq:wagner_sumrate_lb} directly becomes the claimed inequality.
Additionally, for every (strict) subset ${\cal A} \subset \{ 1, 2, \ldots, M \},$
substitute the bound~\eqref{generalized-oohama-bound} into~\eqref{eq:wagner_sumrate_lb} to obtain
\begin{align}
\sum_{i \in {\cal A}}R_i &\geq   R_X(D) \nonumber \\
&  - \frac{1}{2} \log \left( \frac{\entp{\YAc}}{\sAc/|{\cal A}^c|} - \entp{X} \sum_{i \in {\cal A}^c} \frac{e^{-2r_i}}{\sigma_{Z_i}^2} \right) +  \sum_{i \in {\cal A}}  r_i.
\end{align}
It is also important to observe that the argument inside the logarithm is guaranteed to be at least $1$ for all non-negative choices of $r_i.$ To see this, note that
\begin{align}
\lefteqn{\frac{\entp{\YAc}}{\sAc}- \frac{\entp{X}}{|{\cal A}^c|}\sum_{i \in {\cal A}^c} \frac{e^{-2r_i}}{\sigma_{Z_i}^2} } \nonumber \\
& \ge \frac{\entp{\YAc}}{\sAc}- \frac{\entp{X}}{|{\cal A}^c|}\sum_{i \in {\cal A}^c} \frac{1}{\sigma_{Z_i}^2} \\
&= \frac{\entp{\YAc}}{\sAc}- \frac{\entp{X}}{\sAc} \\
& \ge \frac{(\entp{X}+\sAc) - \entp{X}}{\sAc} = 1,
\end{align}
where the last step is due to the entropy power inequality $\entp{\YAc} \ge \entp{X}+\entp{\ZAc},$ and the fact that $\entp{\ZAc} = \sAc.$
Hence, the expression inside the logarithm is lower bounded by $|{\cal A}^c|.$

This has to hold simultaneously (that is, for a fixed $r_i = I(Y_i;U_i|X,W,T)$) for all subsets ${\cal A.}$
This implies that if $(R_1, R_2, \ldots, R_M) \in {\cal R}_{CEO}(D),$
then there must exist non-negative numbers $r_i$ such that the above inequalities are satisfied for all choices of ${\cal A.}$
\end{proof}

\subsection{Proof of Theorem~\ref{thm-AWGN-CEO-generallowerbound} and of the lower bound in Corollary~\ref{thm-AWGN-CEO}}\label{proof:thm-AWGN-CEO-lower}

In this subsection, we leverage Theorem~\ref{thm-CEO-aux} to establish the bounds of Theorem~\ref{thm-AWGN-CEO-generallowerbound} and Corollary~\ref{thm-AWGN-CEO}, {\it i.e.,} the case of equal noise variances. For that case, we relax Theorem~\ref{thm-CEO-aux} to only include the empty set (that is, ${\cal A} = \emptyset$)
and the complete set (that is, ${\cal A} = \{1, 2, \dots, M \}$).
Specifically, for ${\cal A} = \emptyset,$ we find
\begin{align}
0  \geq & R_X(D)    - \frac{1}{2} \log \left( \frac{M\entp{\YM}}{\sigma_{Z}^2}- \frac{\entp{X}}{\sigma_{Z}^2}\sum_{i=1}^M {e^{-2 r_i}} \right). 
\end{align}
Equivalently,
\begin{align}
\frac{1}{ \frac{1}{M}\sum_{i =1}^M {e^{-2 r_i}} } & \geq  \frac{M\entp{X}}{M\entp{\YM}  -  \sigma_{Z}^2 e^{2 R_X(D)}}. \label{eq-aux-proof}
\end{align}
From Jensen's inequality, we have
\begin{align}
\frac{1}{M}\sum_{i=1}^M r_i & \geq   \frac{1}{2}\log \frac{1}{ \frac{1}{M}\sum_{i =1}^M {e^{-2 r_i}} } .
\end{align}
Restricting attention to those values of $D$ for which $M\entp{\YM}  -  \sigma_{Z}^2 e^{2 R_X(D)} \ge 0$ ensures that the denominator on the right hand side of~\eqref{eq-aux-proof} is non-negative. Thus, for such values of $D,$ we find (for ${\cal A} = \emptyset$) that
\begin{align}
\sum_{i=1}^M r_i & \geq  \frac{M}{2} \log^+ \frac{M\entp{X}}{M\entp{\YM}  -  \sigma_{Z}^2 e^{2 R_X(D)}}. \label{eq-subsetchoice-empty}
\end{align}
For ${\cal A} = \{1, 2, \dots, M \},$ we have
\begin{align}
\sum_{i =1}^M R_i  & \geq  R_X(D)  + \sum_{i =1}^M r_i . \label{eq-subsetchoice-full}
\end{align}
Since, by Theorem~\ref{thm-CEO-aux}, there must exist non-negative real numbers $\{r_1, r_2, \ldots, r_M\}$
such that conditions~\eqref{eq-subsetchoice-empty} and~\eqref{eq-subsetchoice-full} are satisfied simultaneously, we conclude
\begin{align}
\sum_{i =1}^M R_i 
  & \geq R_X(D) + \frac{M}{2} \log^+ \frac{M\entp{X}}{M\entp{\YM}  -  \sigma_{Z}^2e^{2 R_X(D)}},
\end{align}
which completes the proof of Theorem~\ref{thm-AWGN-CEO-generallowerbound}.
The lower bound in Corollary~\ref{thm-AWGN-CEO} then follows directly by lower bounding the terms $R_X(D)$ using the lower bound in Inequality~\eqref{eq:direct_lower_bd}.

\subsection{Achievability results for the AWGN CEO Problem (Corollaries~\ref{Cor-CEO-aux-MMSE} and~\ref{thm-AWGN-CEO})}\label{proof:thm-AWGN-CEO-upper}

The achievability results mentioned in this paper all follow from the Berger-Tung region~\cite{berger-msc_lecturenotes,tungthesis-77}.
While these results were originally for discrete memoryless sources and bounded distortion measures, they have been extended to abstract alphabets and suitably smooth distortion functions \cite{housewright-77}, including mean-squared error~\cite{wyner-ic78:wynerziv}.
A detailed analysis for the case of the quadratic Gaussian CEO problem is given in the work of Oohama~\cite{oohama-it98,oohama-it05}. This analysis directly extends to the case of the AWGN CEO problem with Gaussian auxiliaries, as we now briefly explain.
Exactly as in~\cite{oohama-it98,oohama-it05}, we consider a random coding argument where the codebooks are drawn via auxiliary random variables $U_i,$ for $i=1, 2, \ldots, M,$ where
\begin{align}
 U_i &= Y_i + V_i, \label{Eq-defn-Ui}
\end{align}
where $Y_i=X+Z_i$ is the noisy observation of encoder $i,$ and $V_i$ is an independent zero-mean Gaussian of variance $\sigma_{V_i}^2.$
The centerpiece of the proof is the so-called Markov lemma~\cite[Lemma~5]{oohama-it98}, whose proof only uses the fact that {\it conditioned}  on the underlying source sequence, the noisy observations and the auxiliary codebook random variables are Gaussian. Clearly, this still applies in our case, even if $X$ is not Gaussian. This leads to the rate region in~\cite[Equation~(6)]{oohama-it05}, which coincides exactly with~\eqref{eq-thm-CEO-ach-Oohama}, establishing a proof sketch for the achievability part of Corollary~\ref{Cor-CEO-aux-MMSE}. (As a side note, we point out that for non-Gaussian sources $X,$ in general, we can find tighter upper bounds by using auxiliaries $U_i$ of a form different from~\eqref{Eq-defn-Ui}, but this is outside the scope of the present paper. Such arguments are developed for different settings, {\it e.g.,} in~\cite{SanderovichSSK:it08,DBLP:journals/corr/AguerriZCS17}.)

In the remainder, we provide an explicit calculation for the achievability result in Corollary~\ref{thm-AWGN-CEO}, that is, for the sum-rate in the case of equal noise variances (see also the arguments in~\cite{chen-jsac03}). Specifically, we have
\begin{equation}
R_X^{CEO}(D) \leq I(\mathbf{Y}; \mathbf{U}), \label{eq_CEO_SUMRATEDIST_inner}
\end{equation}
where $\mathbf{U} = (U_1, U_2, \ldots, U_M),$ as above. If all noise variances are equal, it is intuitive that a good choice is to also set all the variances $\sigma_{V_i}^2=\sigma_V^2$ to be equal. Then, following standard arguments (see e.g.~\cite[Section~IV]{oohama-it05}), the corresponding distortion is no larger than
\begin{align}
  \min_g {\mathbb E}\left[ \left( X-g(\mathbf{U})\right)^2 \right] & \le  \frac{\sigma_X^2(\sigma_Z^2 + \sigma_V^2)}{M \sigma_X^2 + \sigma_Z^2 + \sigma_V^2},
\end{align}
where the right hand side is the distortion incurred by the optimal linear estimator.
Hence, choosing $\sigma_V^2$ such that the right hand side of this equation equals $D$ characterizes a valid distribution.
It remains to upper bound the corresponding value of $I(\mathbf{Y}; \mathbf{U}).$
\begin{align}
  I(\mathbf{Y}; \mathbf{U}) & =   h({\bf U}) - h({\bf U}|{\bf Y}) \\
    & = h({\bf U}) - \frac{M}{2} \log (2\pi e) \sigma_V^2.
\end{align}
The proof is completed by a standard maximum entropy upper bound on the term $h({\bf U}).$
Specifically, observe that the covariance matrix of the vector $h({\bf U})$ is the $M\times M$ matrix
with entries $\sigma_X^2 + \sigma_Z^2 + \sigma_V^2$ on the diagonal and $\sigma_X^2$ everywhere else.
From a standard maximum entropy argument (subject to a fixed covariance matrix),
we thus find
\begin{align}
 { h({\bf U})} & \le  \frac{1}{2} \log (2 \pi e)^M \left( \left( M\sigma_X^2 + \sigma_Z^2 + \sigma_V^2\right) \left(\sigma_Z^2 + \sigma_V^2 \right)^{M-1} \right) . 
\end{align}
Thus,
\begin{align}
  I(\mathbf{Y}; \mathbf{U}) &\le   \frac{1}{2} \log \left( \frac{ \left( M\sigma_X^2 + \sigma_Z^2 + \sigma_V^2\right) \left(\sigma_Z^2 + \sigma_V^2 \right)^{M-1}}{ \sigma_V^{2M}} \right)  \\
  & =  \frac{1}{2} \log \left( \frac{  M\sigma_X^2 + \sigma_Z^2 + \sigma_V^2}{\sigma_Z^2+\sigma_V^2} \cdot \frac{ \left(\sigma_Z^2 + \sigma_V^2 \right)^{M}}{ \sigma_V^{2M}} \right). 
\end{align}  
Recall that $\sigma_V^2$ is chosen such that we have $D\le \sigma_X^2(\sigma_Z^2 + \sigma_V^2)/(M \sigma_X^2 + \sigma_Z^2 + \sigma_V^2),$
thus,
 \begin{align}
  I(\mathbf{Y}; \mathbf{U}) &\le  \frac{1}{2} \log^+ \frac{\sigma_X^2}{D} + \frac{M}{2} \log \frac{\sigma_Z^2 + \sigma_V^2 }{ \sigma_V^{2}},
\end{align}
and finally, we note that the relation $D\le\frac{\sigma_X^2(\sigma_Z^2 + \sigma_V^2)}{M \sigma_X^2 + \sigma_Z^2 + \sigma_V^2}$
can be rewritten equivalently as
\begin{align}
 \frac{\sigma_Z^2 + \sigma_V^2 }{ \sigma_V^{2}} & \le   \frac{ M \sigma_X^2 }{ M \sigma_X^2 + \sigma_Z^2 - \frac{\sigma_X^2}{D}\sigma_Z^2},
\end{align}
which completes the explicit proof of Inequality~\eqref{eq:ceo_rd_upperbd_sqer} in Corollary~\ref{thm-AWGN-CEO}.

\subsection{Proof of Equations~\eqref{Eq-limit-upperlower-constantD} and~\eqref{Eq-gap-upperlower}}\label{App-Eq-gap-upperlower}

Recall the definition of $\YM$ from Equation~\eqref{eq-def-YM}, namely,
\begin{align}
\YM & = X + \ZM,
\end{align}
and recall that $\ZM$ is a zero-mean Gaussian, independent of $X$, of variance $\sigma_Z^2/M.$
We can leverage~\cite{costa-it85} where it is proved that $\entp{X + \ZM}$ is a concave function of the variance of $\ZM,$ that is, of $\sigma_Z^2/M.$
Therefore, using the definition given in Equation~\eqref{Eq-def-kappaX}, we have
\begin{align}
\entp{\YM} = \entp{X + \ZM} & \le \entp{X} + \kappa_X(\sigma_Z^2/M). \label{Eq-entpYM-upperbd-kappa}
\end{align}
For both Equations~\eqref{Eq-limit-upperlower-constantD} and~\eqref{Eq-gap-upperlower}, we use Equation~\eqref{Eq-entpYM-upperbd-kappa} to lower bound $R_{X, lower}^{CEO}(D).$ Moreover, in the formula for $R_{X, lower}^{CEO}(D),$ we change both $\log^+$ to $\log,$ which cannot increase their values.
For $R_{X,upper}^{CEO}(D),$ we recall that $\sigma_{\YM}^2 = \sigma_X^2 + \sigma_Z^2/M.$

Now, for Equation~\eqref{Eq-limit-upperlower-constantD}, we observe that if $D<\sigma_X^2$ and $M$ is large enough, then the arguments of both logarithms in $R_{X,upper}^{CEO}(D)$ are larger than one. Hence, for such $M,$ we can upper bound the difference as
\begin{align}
\lefteqn{R_{X,upper}^{CEO}(D) - R_{X, lower}^{CEO}(D) } \nonumber \\
& \le \frac{1}{2}\log \frac{\sigma_X^2}{\entp{X}}  + \frac{M}{2}\log \frac{M\sigma_X^2(M\entp{X} + \kappa_X\sigma_Z^2-\frac{\entp{X}}{D}\sigma_{Z}^2)}{M \entp{X}(M\sigma_X^2 + \sigma_Z^2 - \frac{\sigma_X^2}{D}{\sigma_{{Z}}^2})} .
\end{align}
Using the standard bound $\log(1+x) \le x$ and letting $M$ tend to infinity gives the claimed formula.

For Equation~\eqref{Eq-gap-upperlower}, we start by noting that since we assume $M > d/\sigma_X^2,$ we have that $D = d/M < \sigma_X^2.$
Hence, the arguments of both logarithms in the formula for $R_{X,upper}^{CEO}(D=d/M)$ are larger than one.
Therefore, we can bound the difference as follows:
\begin{align}
\lefteqn{R_{X,upper}^{CEO}(d/M) - R_{X, lower}^{CEO}(d/M)} \nonumber \\
 &\le \frac{1}{2}\log \frac{\sigma_X^2}{\entp{X}} + \frac{M}{2}\log \frac{\sigma_X^2(\entp{X} + \kappa_X\sigma_Z^2/M-\frac{\entp{X}}{d}\sigma_{Z}^2)}{ \entp{X}(\sigma_X^2 + \sigma_Z^2/M - \frac{\sigma_X^2}{d}{\sigma_{{Z}}^2})} \nonumber \\
& \le \frac{1}{2}\log \frac{\sigma_X^2}{\entp{X}} + \frac{M}{2}  \frac{ (\kappa_X\sigma_X^2-\entp{X})\sigma_Z^2}{M \sigma_X^2 \entp{X} (1-\sigma_Z^2/d) + \sigma_Z^2}  
\end{align}
where we used the standard bound $\log(1+x) \le x.$
Further upper bounding by dropping the trailing $\sigma_Z^2$ in the denominator establishes Equation~\eqref{Eq-gap-upperlower}.

\section{Proof of Theorem~\ref{lemma-priceofdistributed}}\label{app-proof-Sec-implications}

For centralized encoding, we note that the scenario is precisely the CEO problem with $M=1$ and with a reduced noise variance of $\sigma_Z^2/M.$
Hence, we may use the upper bound in Inequality~\eqref{eq-EPI-explicit-upper} to find
\begin{align}
R_{X}^{R}(D)
 & \leq 
\frac{1}{2}\log^+\frac{\sigma_X^2}{D}  + \frac{1}{2}\log^+ \frac{\sigma_X^2}{\sigma_{\YM}^2 - \frac{\sigma_X^2\sigma_Z^2}{M D}}.
\end{align}
Parameterizing the distortion $D$ as follows
\begin{align}
D = D_\alpha &= \alpha \frac{\sigma_X^2\sigma_Z^2}{M\sigma_{\YM}^2},
\end{align}
we can express the upper bound as
\begin{align}
R_{X}^{R}(D_\alpha)
 & \leq \frac{1}{2}\log^+ \frac{M\sigma_{\YM}^2}{\alpha\sigma_Z^2} +
\frac{1}{2}\log^+\left( \frac{\sigma_X^2}{\sigma_{\YM}^2 (1-\frac{1}{\alpha})}\right). 
\end{align}
As long as $\alpha < M \frac{\sigma_X^2}{\sigma_Z^2}+1,$ this can be combined into
\begin{align}
R_{X}^{R}(D_\alpha)
 & \leq 
\frac{1}{2}\log\left( \frac{M\sigma_X^2}{\sigma_Z^2  (\alpha-1)}\right),
\label{eq-rate-upper-centralized}
\end{align}
where we note that the argument inside the logarithm is lower bounded by one under the stated conditions on $\alpha.$
For distributed (CEO) encoding, we evaluate the lower bound in Corollary~\ref{thm-AWGN-CEO}.
By analogy, we parameterize $\tilde{D}_\beta = \beta \entp{X}\sigma_Z^2/(M\entp{\YM}).$ 
Specifically, we will choose $\beta = \alpha \sigma_X^2/\entp{X},$ which ensures $\tilde{D}_\beta \ge D_\alpha$ (as well as $\beta > 1$).
With such $\beta,$ we get
\begin{align}
\lefteqn{R_{X}^{CEO}(D_\alpha)  \ge R_{X}^{CEO}(\tilde{D}_\beta) } \nonumber \\
& \ge \frac{1}{2} \log^+ \frac{M\entp{\YM}}{\beta \sigma_{Z}^2}   + \frac{M}{2}
\log^+ \left( \frac{ \entp{X}}{\entp{\YM} (1-\frac{1}{\beta} )} \right). \label{eq-rate-lower-distributed}
\end{align}
We further lower bound this by using $\entp{\YM} \ge \entp{X}$ in the first logarithm.
In the second logarithm, we use the assumption that $\kappa_X < \infty,$ which implies via Equation~\eqref{Eq-entpYM-upperbd-kappa}
that\footnote{Alternatively, {\it without}  the assumption that $\kappa_X<\infty,$ we could upper bound as $\entp{\YM} \le \sigma_{\YM}^2 = \sigma_X^2 + \frac{1}{M} \sigma_Z^2.$ This leads to a weaker, but not vacuous bound.}
 $\entp{\YM}\le \entp{X} + \kappa_X(\sigma_Z^2/M).$
We plug in $\beta  = \alpha \sigma_X^2/\entp{X}.$
Moreover, we further lower bound the first logarithm by changing the $\log^+$ to $\log,$ which thus leads to
\begin{align}
R_{X}^{CEO}(D_\alpha)  & \ge \frac{1}{2} \log \frac{M(\entp{X})^2}{\alpha \sigma_X^2 \sigma_{Z}^2} \nonumber \\
& + \frac{M}{2} \log^+ \left( \frac{ \alpha \sigma_X^2 \entp{X}}{(\entp{X} + \frac{\kappa_X}{M} \sigma_Z^2) (\alpha \sigma_X^2 - \entp{X} )} \right).
\end{align}
Subtracting Inequality~\eqref{eq-rate-upper-centralized} from this expression gives the claimed formula.

\section{Proof of Theorem~\ref{thm-jscc}}\label{App-thm-jscc}

We start by proving Inequality~\eqref{Eq-scalinglaw-digital}.
Let us loosen the lower bound in Corollary~\ref{thm-AWGN-CEO} (Inequality~\eqref{eq:ceo_rd_lowerbd_sqer}) to
\begin{align}
R_X^{CEO}(D) &\geq \frac{M}{2} \log^+ \frac{M \entp{X}}{M\entp{\YM}-\frac{\entp{X}}{D}\sigma_{Z}^2 } . 
\end{align}
Next, we set this lower bound equal to $R,$ the total communication (sum) rate available over the multiple-access channel.
Then, we find
\begin{align}
D & \ge  \frac{\entp{X}\sigma_{Z}^2}{ M\entp{X}(1-\exp(-2R/M))+ M(\entp{\YM}-\entp{X})}.
\end{align}
To simplify further, we observe that
$M(1-\exp(-2R/M)) \le 2R,$ hence,
\begin{align}
D & \ge  \frac{\entp{X}\sigma_Z^2}{2 \entp{X} R + M(\entp{\YM}- \entp{X})}.
\end{align}
To bound the total rate $R$ available over the multiple-access channel is a somewhat subtle issue due to the fact that the messages produced by the source code may be dependent. Here, we merely bound this total rate by the rate for a corresponding vector (or ``multiple-antenna'') channel, where the channel input is thus a vector of length $M$ (and the channel output is a scalar, equal to the sum of the $M$ inputs plus noise).
For such a system, it is well known that
\begin{align}
R & \le  \frac{1}{2} \log \left( 1 + \frac{M^2 P}{\sigma^2_{channel}} \right),
\end{align}
where $P$ is the transmit power per user on the multiple-access channel.
Hence, for a digital strategy in the sense discussed here, the resulting distortion is at least
\begin{align}
D_{\mathrm{d}} & \ge  \frac{\entp{X}\sigma_Z^2}{ \entp{X}\log \left( 1 + \frac{M^2 P}{\sigma^2_{channel}} \right) + M(\entp{\YM}-\entp{X})}.
\end{align}
Finally, using the upper bound from Equation~\eqref{Eq-entpYM-upperbd-kappa}, we find that $M(\entp{\YM}-\entp{X}) \le  \kappa_X\sigma_Z^2,$
which completes the proof of Inequality~\eqref{Eq-scalinglaw-digital}.

For Equation~\eqref{Eq-scalinglaw-analog}, we proceed as follows:
At time $i$, encoder $m$ sets $U_{m}[i] =  \sqrt{\frac{P}{\sigma_X^2+\sigma_Z^2}}Y_{m}[i].$
Clearly, this satisfies the constraint in Equation~\eqref{Eq-PowerCon} and is thus a valid strategy.
At the receiver, we set $\hat{X}[i]$ equal to the linear mean-squared error estimator of $X[i]$ based on $V[i],$
which is well known to be
\begin{align}
 \hat{X}[i] & =  \frac{\sqrt{\frac{P}{\sigma_X^2+\sigma_Z^2}}M \sigma_X^2}{\frac{P}{\sigma_X^2+\sigma_Z^2}(M^2\sigma_X^2+M\sigma_Z^2) + \sigma_{channel}^2} V[i] .
\end{align}
A direct calculation reveals that for each $i,$ we have
\begin{align}
    \lefteqn{ {\mathbb E}\left[ ( X[i] - \hat{X}[i] )^2 \right] } \nonumber \\
    & =   \frac{ \sigma_X^2 \sigma_N^2}{M \sigma_X^2 + \sigma_N^2 }
\left( 1 +  \frac{ M ( \sigma_X^2\sigma_{channel}^2/ \sigma_N^2 )  }{\frac{M\sigma_X^2+\sigma_N^2}{\sigma_X^2+\sigma_N^2} MP +\sigma_{channel}^2} \right),
\end{align}
which thus establishes Equation~\eqref{Eq-scalinglaw-analog}.

\bibliographystyle{IEEEtran}

\end{document}